\documentclass[a4paper,11pt]{article}

\usepackage{graphicx,amssymb,amsmath,amsfonts}
\usepackage{amsthm}
\usepackage{hyperref}
\usepackage{color}
\usepackage{fullpage}
\usepackage[T1]{fontenc}
\usepackage{url}
\usepackage{framed}
\usepackage{boxedminipage}
\usepackage{enumerate}
\usepackage{cite}
\usepackage[lined,boxed,commentsnumbered,ruled,vlined,noend,linesnumbered,boxed]
{algorithm2e}

\newtheoremstyle{mytheorem}{3pt}{3pt}{\slshape}{}{\bfseries}{}{.5em}{}
\theoremstyle{mytheorem}
\newtheorem{lemma}{Lemma}
\newtheorem{theorem}{Theorem}

\newtheorem{observation}{Observation}

\theoremstyle{definition}

\def\denseitems{
    \itemsep1pt plus1pt minus1pt
    \parsep0pt plus0pt
    \parskip0pt\topsep0pt}
    
\newbox\ProofSym \setbox\ProofSym=\hbox{%
  \unitlength=0.18ex%
  \begin{picture}(10,10) \put(0,0){\framebox(9,9){}}
    \put(0,3){\framebox(6,6){}}
  \end{picture}}

\title{Maximum-Width Empty Square and Rectangular Annulus\thanks{%
S. W. Bae was supported by Basic Science Research Program through the National
Research Foundation of Korea (NRF) funded
by the Ministry of Education (2018R1D1A1B07042755).
P. R. S. Mahapatra was supported by Research Project through Department of Atomic Energy (NBHM),
Goverment of India with Ref. No. 2/48(19)/2014/R\&D-II/1045. 
}
}

\author{%
Sang Won Bae\thanks{%
Division of Computer Science and Engineering,
Kyonggi University, Suwon, Republic of Korea.
Email: \texttt{swbae@kgu.ac.kr}.
}
\and
Arpita Baral\thanks{%
Department of Computer Science and Engineering,
University of Kalyani, India.
Emails: \texttt{arpitabaral@gmail.com}, \texttt{priya@klyuniv.ac.in}.
}
\and
Priya Ranjan Sinha Mahapatra\footnotemark[3]
}

\begin{document}
\maketitle

\begin{abstract}
An annulus is, informally, a ring-shaped region, often described by
two concentric circles.
The maximum-width empty annulus problem asks to find
an annulus of a certain shape with the maximum possible width
that avoids a given set of $n$ points in the plane.
This problem can also be interpreted as the problem of finding an optimal
location of a ring-shaped obnoxious facility among the input points.
In this paper, we study square and rectangular variants
of the maximum-width empty anuulus problem,
and present first nontrivial algorithms.
Specifically, our algorithms run in $O(n^3)$ and $O(n^2 \log n)$ time
for computing a maximum-width empty axis-parallel square and rectangular annulus,
respectively.
Both algorithms use only $O(n)$ space.
\end{abstract}

\section{Introduction}

The problem of computing a minimum-size geometric object
that encloses an input point set $P$ is
one of the central research problems in computational geometry.
This type of problem has been extensively studied
with direct applications to location of desirable facilities to customers $P$,
for a variety of different geometric shapes
including circles~\cite{ps-cgi-90}, rectangles~\cite{tou-sgprc-83},
and annuli~\cite{ast-apsgo-94,as-eago-98,ght-oafepranw-09,jmkd-mwra-2012,AHIMPR-03,bae-cmwea-2018}.


On the other hand, in some applications, the facility to be built among $P$
is considered \emph{obnoxious}, that is, every member in $P$ wants to
be as far away from it as possible.
The problem of locating an obnoxious facility is often interpreted as
the problem of finding a maximum-size empty geometric object among $P$.
For examples, the center of a largest circle or square that is empty of $P$
corresponds to an optimal location of a point obnoxious facility
that maximizes the Euclidean or $L_\infty$ distance, respectively,
from its closest point in $P$.
A largest empty circle or square can be found in optimal $O(n \log n)$ time
using the Voronoi diagram~\cite{t-cleclc-83},
and the best known algorithm that computes an empty axis-parallel rectangle of maximum area
runs in $O(n \log^2 n)$ time by Aggarwal and Suri~\cite{as-facler-87}.
The widest empty corridor problem,
in which one wants to find a widest empty strip among $P$ of arbitrary orientation,
is another interesting problem in this concept.
After Houle and Maciel~\cite{hm-fwecsp-88} presented an $O(n^2)$-time algorithm
for this problem, a lot of variants and extensions have been addressed,
including 
the widest L-shaped corridor problem~\cite{c-welsc-96},
and the widest $1$-corner corridor problem~\cite{dls-fwe1cc-06}.
Note that these problems are equivalent to those of finding an optimal location
of an obnoxious facility whose shape is of a line, a line segment, or a polygonal chain.

In this paper, along this line of research, we study
the \emph{maximum-width empty annulus problem}.
Informally, an annulus is a ring-shaped region, often described by
two concentric circles.
Thus, the maximum-width empty annulus problem is to find
an optimal location of a ring-shaped obnoxious facility among the input points $P$.
Specifically, we discuss its square and rectangular variants,
and present first nontrivial algorithms.
Our algorithms run in $O(n^3)$ and $O(n^2 \log n)$ time
for computing a maximum-width axis-parallel square and rectangular annulus, respectively,
that is empty of a given set $P$ of $n$ points in the plane.
Both algorithms use only $O(n)$ space.

There has been a little work on the maximum-width empty annulus problem.
D{\'{\i}}az{-}B{\'{a}}{\~{n}}ez et al.~\cite{dhmrs-leap-03} first studied the problem
for circular annulus, and proposed an $O(n^3\log n)$-time and $O(n)$-space
algorithm to solve it.
To our best knowledge, there was no known correct algorithm in the literature
for the maximum-width empty square or rectangular annulus problem.
Mahapatra~\cite{mahapatra-larea-2012} considered the maximum-width empty
rectangular annulus problem and claimed an incorrect $O(n^2)$-time algorithm.
There is a missing argument in Observation~2 of \cite{mahapatra-larea-2012},
which incorrectly claimed that the total number of potential outer rectangles
forming an empty rectangular annulus is $n-1$.

Unlike the maximum-width empty annulus problem,
the problem of finding a minimum-width annulus that encloses $P$ has recently
attained intensive interests from researchers.
As a classical one, circular annuli have been studied earlier
with applications to the roundness problem~\cite{w-nmpp-86,rz-epccmrsare-92,efnn-rauvd-89}, and the currently best known algorithm runs
in $O(n^{3/2+\epsilon})$ time~\cite{ast-apsgo-94,as-eago-98}.
Computing a minimum-width axis-parallel square or rectangular annulus that encloses $n$ points $P$
can be done in $O(n \log n)$ or $O(n)$ time, respectively~\cite{ght-oafepranw-09,AHIMPR-03}.
Mukherjee et al.~\cite{jmkd-mwra-2012} considered the problem of
identifying a rectangular annulus of minimum width that encloses $P$ in arbitrary orientation,
and presented an $O(n^2 \log n)$-time algorithm.
Bae~\cite{bae-cmwea-2018} studied a minimum-width square annulus in arbitrary orientation
and showed that it can be solved in $O(n^3 \log n)$ time.

The rest of the paper is organized as follows:
In Section~\ref{sec:problemdef}, we introduce some definitions and notations,
and precisely define our problems.
Our algorithms are described in the following sections:
Section~\ref{sec:sq} for computing a maximum-width empty axis-parallel square annulus
and Section~\ref{sec:rect} for computing a maximum-width empty axis-parallel rectangular annulus.
Finally, we finish the paper with concluding remarks in Section~\ref{sec:conclusions}.

\section{Problem Definition and Terminologies}\label{sec:problemdef}

Throughout the paper, we consider a Cartesian coordinate system
of the plane $\mathbb{R}^2$ with the $x$- and $y$-axes.
For any point $p$ in the plane $\mathbb{R}^2$,
we denote by $x(p)$ and $y(p)$ its $x$- and $y$-coordinates.
For an axis-parallel rectangle or square, its four sides are naturally identified by \emph{top}, \emph{bottom}, \emph{left}, and \emph{right} sides, respectively.

For an axis-parallel square, the intersection point of its two diagonals
is called its \emph{center}, and its \emph{radius} is half its side length.
An \emph{axis-parallel square annulus} is the region between two concentric
axis-parallel squares $S$ and $S'$, where $S' \subseteq S$.
We call $S$ and $S'$ the \emph{outer} and \emph{inner} squares, respectively,
of the annulus.
The \emph{width} of a square annulus is defined to be
the difference of radii of its outer and inner squares.
See \figurename~\ref{fig:annulus}(left) for an illustration.

 \begin{figure}[t]
 \centering
 \includegraphics[width=.7\textwidth]{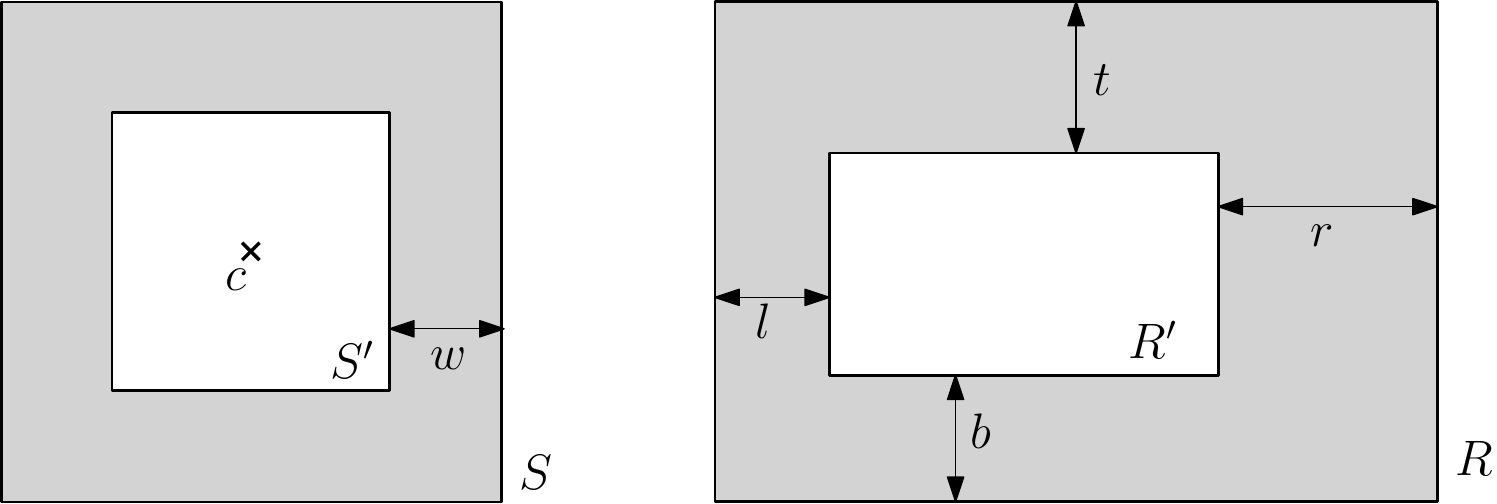}
 \caption{A square annulus of width $w$
 with outer and inner squares $S$ and $S'$ having a common center $c$ (left)
 and a rectangular annulus with outer and inner rectangles $R$ and $R'$
 whose top-, bottom-, left-, right-widths are $t$, $b$, $l$, and $r$,
 respectively (right).}
 \label{fig:annulus}
 \end{figure}

An \emph{axis-parallel rectangular annulus} is the region obtained by subtracting the interior of an axis-parallel rectangle $R'$ from another axis-parallel rectangle $R$ such that $R' \subseteq R$.
We call $R$ and $R'$ the \emph{outer rectangle} and \emph{inner rectangle} of
the annulus, respectively.
Consider a rectangular annulus $A$ defined by its outer and inner rectangles, $R$ and $R'$.
By our definition, note that $R$ and $R'$ defining annulus $A$ do not have to be concentric, so that $A$ may not be a symmetric shape.
The \emph{top-width} of $A$ is the vertical distance between the top sides of $R$ and $R'$, and the \emph{bottom-width} of $A$ is the vertical distance between their bottom sides.
Analogously, the \emph{left-width} and \emph{right-width} of $A$ are defined to be the horizontal distance between the left sides of $R$ and $R'$ and the right sides of $R$ and $R'$, respectively.
Then, the \emph{width} of $A$ is defined to be the minimum of the four values: the top-width, bottom-width, left-width, and right-width of $A$.
See \figurename~\ref{fig:annulus}(right) for an illustration.

In this paper, we only discuss squares, rectangles, square annuli,
and rectangular annuli that are axis-parallel.
Hence, we shall drop the term ``axis-parallel'', and
any square, rectangle, or annulus we discuss is assumed to be axis-parallel. 

Let $P$ be a set of $n$ points in $\mathbb R^{2}$.
A square or rectangular annulus $A$ is said to be \emph{empty} of $P$,
or just \emph{empty} when there is no confusion,
if the interior of $A$ does not contain any point in $P$.
Consider any empty square or rectangular annulus $A$.
Then, $A$ induces a partition of $P$ into
two subsets $P_\mathrm{out}$ and $P_\mathrm{in}$
such that $P_\mathrm{in}$ is the set of points in $P$ lying in the interior
or on the boundary of the inner square or rectangle of $A$, and
$P_\mathrm{out} = P \setminus P_\mathrm{in}$.
If both $P_\mathrm{out}$ and $P_\mathrm{in}$ are nonempty,
then we say that $A$ is \emph{valid}.
In this paper, we address the following problems:

\begin{center}
\noindent\framebox{\begin{minipage}{4.50in}	
 {\textsc{MaxWidthEmptySquareAnnulus}} (MaxESA)\\
\emph{Input}: A set of points $P$ in $\mathbb R^{2}$ \\
\emph{Output}: A valid empty square annulus $A$ of maximum width
\end{minipage}}
\end{center}

\begin{center}
\noindent\framebox{\begin{minipage}{4.50in}	
 {\textsc{MaxWidthEmptyRectangularAnnulus}} (MaxERA)\\
\emph{Input}: A set of points $P$ in $\mathbb R^{2}$ \\
\emph{Output}: A valid empty rectangular annulus $A$ of maximum width
\end{minipage}}
\end{center}

The constraint that the resulting empty annulus should be valid is essential
to make the problem nontrivial;
the same constraint has often been considered in the problem of
computing empty objects of maximum size~\cite{dhmrs-leap-03,c-welsc-96,jp-wcp-94,dls-fwe1cc-06}.
Throughout the paper, we are interested only in valid empty annuli,
so we shall drop the term ``valid'' unless stated otherwise.


\section{Maximum-Width Empty Square Annulus} \label{sec:sq}

In this section, we present an algorithm that computes
a maximum-width valid empty square annulus for a given set $P$ of $n$ points.

Consider any empty square annulus $A$.
Keeping the same partition of $P$ by $A$, one can enlarge the outer square
and shrink the inner square so that some points of $P$ lie on
the boundary of the outer and inner squares.
This process implies the following observation.
A side of a rectangle or a square is said to be \emph{at infinity}
if it is a translated copy of a line segment by a translation vector at infinity.
\begin{observation} \label{obs:sq_conf}
 There exists a maximum-width empty square annulus such that
 one side of its inner square contains a point of $P$ and
 one of the following holds:
 (i) There are a pair of opposite sides of its outer square,
  each of which contains a point of $P$,
 (ii) there are two adjacent sides of its outer square,
  each of which contains a point of $P$,
   and the other two sides are at infinity, or
 (iii) One side of its outer square contains a point of $P$ and
  the other three sides are at infinity.
\end{observation}
\begin{proof}
Consider any maximum-width empty square annulus $A$ with
outer square $S$, inner square $S'$, and center $c$.
If there is no point on the boundary of the inner square $S'$,
then we can shrink it with the same center, since $A$ is valid,
so that we can strictly increase its width, a contradiction.
So, there must be at least one point lying on the boundary of $S'$.
Similarly, if $S$ contains no point of $P$ on its boundary,
then we can enlarge it with the same center, keeping $S'$ the same,
so that we can strictly increase its width, a contradiction.
So, there must be at least one point $p \in P$ lying on the boundary of $S$.

Without loss of generality, we assume that $p$ lies on the top side
of $S$.
If the other three sides of $S$, except the top side, contain no point of $P$,
then we move the center $c$ in the direction away from $p$,
keeping $p$ on the top side of $S$ and the width of $A$ unchanged,
until one of the other three sides of $S$ hits a point of $P$.
If this enlarging process does not stop, then we end with
the three sides of $S$ being at infinity,
resulting in case (iii).
Otherwise, we now have another side of $S$ containing a point $p' \in P$.
There are two cases: either it is the bottom side of $S$,
which is now in case (i),
or is the left or right side of $S$.

In the latter case, we have two adjacent sides of $S$ each of which contains
a point of $P$.
Without loss of generality, we assume the second side is the left side.
Now, we again move the center in the direction away
from the top-left corner of $S$,
keeping $p$ on the top side of $S$, $p'$ on the left side of $S$,
and the width of $A$ unchanged,
until one of the bottom and right sides hits a point of $P$.
If this enlarging process does not stop, then we end with
the two sides of $S$ being at infinity,
resulting in case (ii).
Otherwise, we now have the third side of $S$ containing a point of $P$,
and this falls in case (i).
\end{proof}

By Observation~\ref{obs:sq_conf}, we now have three different configurations
of empty square annuli to search for.
If this is case (iii), then observe that it corresponds to
a maximum-width empty horizontal or vertical strip,
which also can be reduced to the problem of finding the maximum gap
in $\{x(p) \mid p\in P\}$ or in $\{y(p) \mid p\in P\}$.
Hence, case (iii) can be handled in $O(n)$ time after sorting $P$.

On the other hand, if this is case (ii), then
the resulting square annulus corresponds to
a maximum-width empty ``axis-parallel'' L-shaped corridor.
It is known that a maximum-width empty L-shaped corridor over all orientations
can be computed in $O(n^3)$ time with $O(n^3)$ space by Cheng~\cite{c-welsc-96},
while we are seeking only for axis-parallel ones.
Here, we give a simple $O(n^2 \log n)$ time algorithm for this problem.
\begin{theorem} \label{thm:L-corridor}
 Given $n$ points in the plane,
 one can compute a widest empty axis-parallel corridor in $O(n^2 \log n)$ time
 using $O(n)$ space.
\end{theorem}
\begin{proof}
Consider the induced grid by points in $P$ obtained by
drawing the horizontal and vertical lines through every point in $P$.
For each grid point $o$, consider the four quadrants at $o$ obtained by
the vertical and horizontal lines through $o$.
Consider the first quadrant, or equivalently the top-right quadrant, $Q$.
By finding out the point with the smallest $x$-coordinate
and the point with the smallest $y$-coordinate among $P\cap Q$,
we can identify a candidate L-shape empty corridor with corner $o$.
This query can be handled in $O(\log n)$ time
by the segment dragging data structure by Chazelle~\cite{c-asdii-88}
which can be built in $O(n \log n)$ time using $O(n)$ space.

Since there are $O(n^2)$ grid points, we can compute
a maximum-width empty axis-parallel corridor in $O(n^2 \log n)$ time.
\end{proof}

Now, we suppose that the solution falls in case (i)
of Observation~\ref{obs:sq_conf},
so that both two opposite sides of its outer square contains a point of $P$.
Without loss of generality, we assume that each of the top and bottom sides
of the outer square of our target annulus contain a point of $P$.
The other case can be handled in a symmetric way.

First, as preprocessing, we sort $P$ in the decreasing order of
their $y$-coordinates,
so $P = \{p_1, p_2, \ldots, p_n\}$, where $y(p_1) \geq \cdots \geq y(p_n)$.
We also maintain the list of points in $P$ sorted in their $x$-coordinates.
Our algorithm runs repeatedly for all pairs of indices $(i, j)$
with $1\leq i < j-1 < n$,
and finds a maximum-width empty square annulus such that
the top and bottom sides of its outer square contain $p_i$ and $p_j$,
respectively.

From now on, we assume $i$ and $j$ are fixed.
Let $P_{ij} := \{p_{i+1}, \ldots, p_{j-1}\}$,
$r := (y(p_i)-y(p_j))/2$,
and $\ell$ be the horizontal line with $y$-coordinate $y(\ell) = (y(p_i)+y(p_j))/2$.
Provided that $p_i$ lies on the top side and $p_j$ lies on the bottom side
of the outer square,
the possible locations of its center is constrained to be on $\ell$.
For a possible center $c\in \ell$,
let $S(c)$ be the square centered at $c$ with radius $r$.
Then, the corresponding inner square $S'(c)$ is determined by center $c$
and the farthest point among those points in $P_{ij}$
lying in the interior of $S(c)$.
Here, the distance is measured by the $L_\infty$ metric.
More precisely, the radius of $S'(c)$ is exactly
$\max_{p\in P_{ij}\cap S(c)} \| p-c\|_\infty$,
where $\|\cdot\|_\infty$ denotes the $L_\infty$ norm,
and we want to minimize this over the relevant segment $C \subset \ell$
such that $S(c)$ for $c\in C$ contains $p_i$ and $p_j$ on its
top and bottom sides.
Note that the length of segment $C \subset \ell$ is
exactly $2r - |x(p_i) - x(p_j)|$.

For the purpose, we define $f_p(c)$ for each $p\in P_{ij}$
and all $c \in \ell$ to be
\[ f_p(c) = \begin{cases}
            \| p - c \|_\infty & \text{ if $\|p - c\|_\infty < r$} \\
            0 & \text{ otherwise}
            \end{cases},
\]
and let $F(c):=\max_{p \in P_{ij}} f_p(c)$ be their upper envelope.
Note that our goal is to minimize the upper envelope $F$ over $C \subset \ell$.

 \begin{figure}[tb]
 \centering
 \includegraphics[width=.8\textwidth]{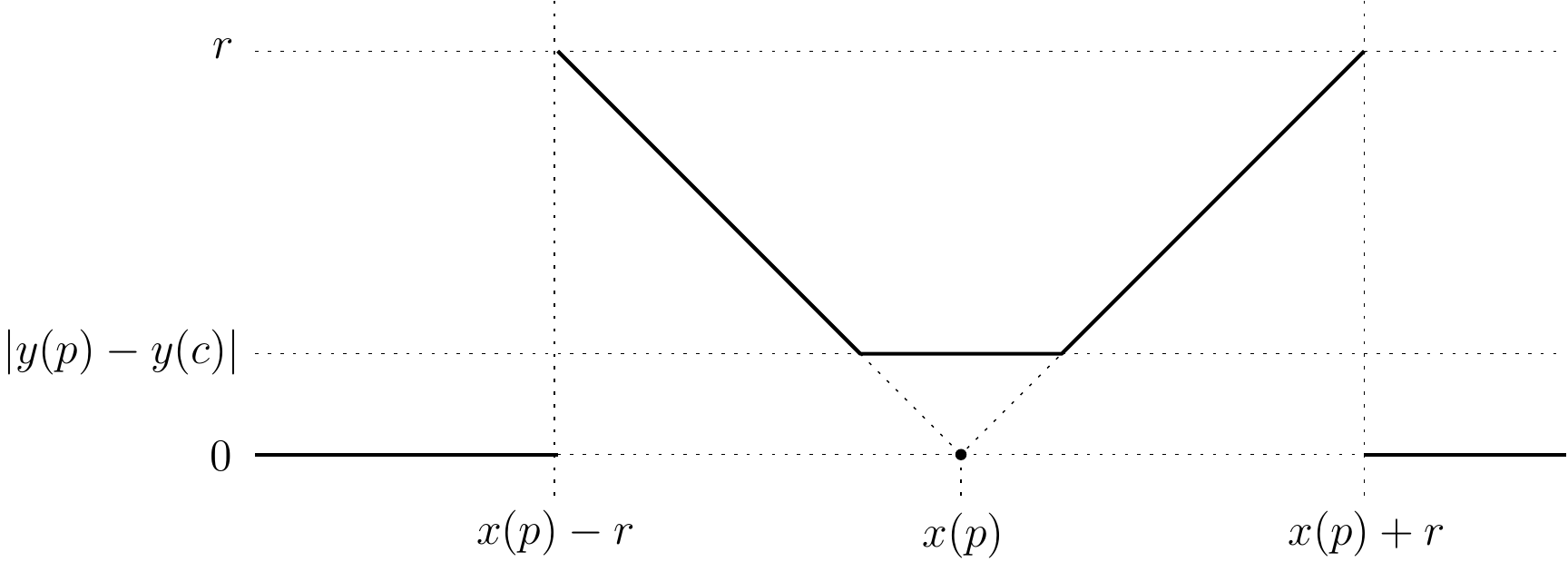}
 \caption{Illustration of the graph of function $f_p$ for $p\in P$.}
 \label{fig:function_f}
 \end{figure}

Note that $x(p)-r < x(c) < x(p)+r$, or equivalently, $p \in S(c)$
if and only if $\|p-c\|_\infty < r$ for any $p\in P_{ij}$.
As also observed in Bae~\cite{bae-cmwea-2018},
the function $f_p$ is piecewise linear with at most three pieces
over $c \in \ell$ such that $p \in S(c)$,
and the three pieces have slopes $-1$, $0$, and $1$
in this order.
Moreover, the height of the part of slope $0$ is exactly
$|y(p) - y(c)| = |y(p) - y(\ell)|$ and
the extensions of the two pieces of slope $-1$ and $1$ always
cross at the point of $x$-coordinate $x(p)$ and height $0$.
See \figurename~\ref{fig:function_f}.
These properties of $f_p$ can be easily verified from
the behavior of the $L_\infty$ norm.
By the above observations,
one can explicitly compute $F$
by computing the upper envelope of $O(n)$ line segments in $O(n\log n)$ time~\cite{h-fuels-89}.
Applying this to all possible pairs $(i, j)$ yields
an $O(n^3 \log n)$ time algorithm.

In the following, we show how to improve this to $O(n^3)$ time
by decomposing the function $f_p$ into two functions $g_p$ and $h_p$.
For each $p\in P_{ij}$ and $c\in \ell$, define
\[ g_p(c) = |x(p) - x(c)| \quad \text{and} \quad h_p(c) = |y(p) - y(c)|
 \quad \text{ if $\|p - c\|_\infty < r$},
\]
and $g_p(c) = h_p(c) = 0$, otherwise.
Also, let 
 \[ G(c) := \max_{p \in P_{ij}} g_p(c) \quad \text{and} \quad
 H(c) := \max_{p\in P_{ij}} h_p(c). \]
As $\|p - c\|_\infty = \max\{|x(p) - x(c)|, |y(p) - y(c)|\}$,
it is obvious that $f_p(c) = \max\{g_p(c), h_p(c)\}$,
and hence $F(c) = \max\{G(c), H(c)\}$.
We now show that the functions $G$ and $H$ can be explicitly computed
in $O(n)$ time.
\begin{lemma} \label{lem:sq_GH}
 The functions $G$ and $H$ can be explicitly computed in $O(n)$ time.
\end{lemma}
\begin{proof}
First recall that we know the list of points in $P$ sorted
in their $x$-coordinates in the preprocessing.
From this list, we can find out the list of points in $P_{ij}$
sorted in the $x$-coordinates in linear time.

By above discussions, the graph of function $g_p$ consists of
exactly two line segments, except of the two parts of zero, of slopes $-1$ and $1$,
and the breakpoint lies at $x(p)$.
That is, all the functions $g_p$ are translates of this V-shaped segments.
We first compute the upper envelope of all segments of slope $-1$ of $g_p$.
This can be done in $O(n)$ time since we know their sorted order.
Analogously, we compute the upper envelope of all segments of slope $1$ of $g_p$
in $O(n)$ time.
Merging these two envelope into the upper envelope $G$ can be done in $O(n)$ time.
Thus, we can compute $G$ in $O(n)$ time.

To compute $H$, observe that the graph of function $h_p$
consists of two zero parts and a horizontal line segment of length $2r$.
Since every non-zero segment in functions $h_p$ for $p\in P_{ij}$
has the same length and we know the sorted order,
we can compute their upper envelope $H$ in $O(n)$ time.
\end{proof}

Since the function $F$ is the upper envelope of $G$ and $H$,
we can compute $F$ in $O(n)$ time using the explicit description
of functions $G$ and $H$.
Note that the three functions $F$, $G$, and $H$ are piecewise linear
with $O(n)$ breakpoints.
Consequently, we can compute $F$ and find a lowest point of $F$
over $C \subset \ell$ in $O(n)$ time,
and hence a maximum-width empty square annulus of case (i) can be found
in $O(n^3)$ time.
Finally, we conclude the following theorem.
\begin{theorem} \label{thm:sq}
 Given $n$ points in the plane,
 a maximum-width empty square annulus can be computed
 in $O(n^3)$ time using $O(n)$ space.
\end{theorem}
\begin{proof}
Recall Observation~\ref{obs:sq_conf} stating that there are three cases
of a possible maximum-width empty square annulus.
In the overall algorithm, we check all the three possibilities
as follows:
For case (i), we run the above algorithm in $O(n^3)$ time.
For case (ii), we compute a maximum-width empty axis-parallel L-shaped corridor
in $O(n^2 \log n)$ time by Theorem~\ref{thm:L-corridor}.
For case (iii), we compute a maximum-width empty axis-parallel strip
in $O(n)$ time.
Thus, the correctness directly follows from Observation~\ref{obs:sq_conf}.
All these routines use $O(n)$ space.
\end{proof}

\section{Maximum-Width Empty Rectangular Annulus} \label{sec:rect}

In this section, we present an algorithm computing a maximum-width
empty rectangular annulus.
First we give several basic observations on maximum-width empty rectangular annuli.

\subsection{Configurations of empty rectangular annuli}

Consider any empty rectangular annulus $A$ and
the partition of $P$ induced by $A$.
As done for square annuli before,
one can enlarge the outer rectangle of $A$ and shrink its inner rectangle,
while keeping the partition of $P$ and not decreasing the width of $A$.
This results in the following observation.

\begin{observation} \label{obs:rect_conf1}
 There exists a maximum-width empty rectangular annulus such that
 each side of its outer rectangle either contains a point of $P$
 or lies at infinity,
 and every side of its inner rectangle contains a point of $P$.
\end{observation}
\begin{proof}
Let $A$ be any empty rectangular annulus
with outer rectangle $R_{out}$ and inner rectangle $R_{in}$.
We first enlarge $R_{out}$ to obtain a new outer rectangle $R'_{out}$ as follows:
If the top side of $R_{out}$ contains no point of $P$,
then we slide it upwards until it hits a point of $P$.
If this process stops at some point, then the top side contains a point of $P$;
otherwise, it is now at infinity.
Afterwards, we perform the same process for the other three sides,
one by one.
Hence, each side of the resulting rectangle $R'_{out}$ either
contains a point of $P$ or lies at infinity.

Next, we shrink $R_{in}$ by sliding each side of $R_{in}$ inwards
until it hits a point of $P$,
and let $R'_{in}$ be the resulting rectangle.
Since $A$ is assumed to be valid,
there is at least one point of $P$ in the interior or on the boundary of $R_{in}$,
which implies that every side of $R'_{in}$ contains a point of $P$.

Now, consider a new rectangular annulus $A'$ defined by
its outer rectangle $R'_{out}$ and inner rectangle $R'_{in}$.
By the above processes, it is obvious that
$A'$ is also empty and the width of $A'$ is not smaller than that of $A$.
Moreover, $A'$ satisfies the condition described in the statement.
By our construction, we show that there always exists such an snnulus $A'$
for any empty rectangular anuulus, so the observation follows.
\end{proof}

In Observation~\ref{obs:rect_conf1}, note that each side of a rectangle is considered to include its endpoints.
Thus, a point $p\in P$ can be contained in two adjacent sides of a rectangle
if $p$ is located at a corner.
See \figurename~\ref{fig:rect_conf}(left).

 \begin{figure}[tb]
 \centering
 \includegraphics[width=.8\textwidth]{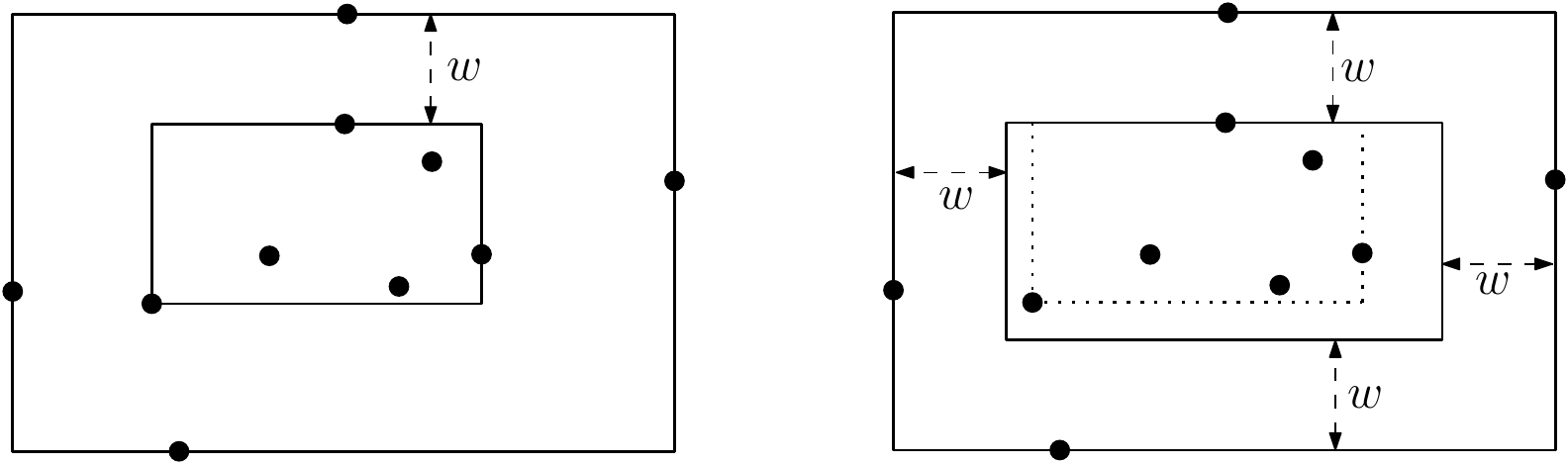}
 \caption{Empty rectangular annuli of maximum width $w$: (left) Each side of the inner and outer rectangles contains a point. (right) A maximum-width empty rectangular annulus that is uniform and also top-anchored.}
 \label{fig:rect_conf}
 \end{figure}

A rectangular annulus $A$ is said to be \emph{width-uniform}, or simply \emph{uniform}, if its top-width, bottom-width, left-width, and right-width are all equal to its width.
In the following observation, we show that we can focus only on uniform rectangular annuli to solve our problem.
\begin{observation} \label{obs:rect_conf_uniform}
 There exists a maximum-width empty rectangular annulus $A$ that is uniform
 such that the following property holds:
 each side of its outer rectangle either contains a point of $P$
 or lies at infinity, and
 at least one side of its inner rectangle contains a point of $P$.
\end{observation}
\begin{proof}
Let $A$ be a maximum-width empty rectangular annulus with the property described in Observation~\ref{obs:rect_conf1}.
Let $R_{out}$ and $R_{in}$ be the outer and inner rectangles of $A$,
and $w$ be the width of $A$.

Without loss of generality, we assume that $w$ is equal to the top-width of $A$, that is, the top-width is the smallest among the other three.
We then enlarge the inner rectangle $R_{in}$ to have a new inner rectangle $R'_{in}$ by sliding the left side to the left, the right side to the right, and the bottom side downwards so that the resulting annulus $A'$ formed by $R_{out}$ and $R'_{in}$ is uniform and the width of $A'$ is equal to $w$, the width of $A$.
See \figurename~\ref{fig:rect_conf} for an illustration.
Since $R_{in} \subseteq R'_{in}$ and thus $A' \subseteq A$,
we conclude that $A'$ is also empty of $P$.
The proof is done by also observing that every side of $R_{out}$ contains at least one point of $P$ and the top side of $R'_{in}$ contains at least one point of $P$ by Observation~\ref{obs:rect_conf1}.
\end{proof}

See \figurename~\ref{fig:rect_conf}(right) for an illustration of Observation~\ref{obs:rect_conf_uniform}.
This observation suggests a specific configuration of annuli for us to solve the problem.
First of all, we do not have to consider non-uniform annuli.
Moreover, candidate outer rectangles are defined by at most four points in $P$.
If we fix an outer rectangle $R$, then the inner rectangle that maximizes the width is also determined by searching points in $P \cap R$.
This already yields an $O(n^5)$-time algorithm for our problem.

Let $A$ be an empty rectangular annulus satisfying the condition described
in Observation~\ref{obs:rect_conf1}.
We call $A$ \emph{top-anchored} (or, \emph{bottom-anchored}, \emph{left-anchored}, \emph{right-anchored})
if both the top sides (or, bottom sides, left sides, right sides, resp.) of the outer and inner rectangles of $A$ contain a point of $P$.
For example, \figurename~\ref{fig:rect_conf}(right) shows an empty top-anchored rectangular annulus.

\begin{observation} \label{obs:rect_conf_anchored}
 There exists a maximum-width empty rectangular annulus $A$ that
 satisfies the condition described in Observation~\ref{obs:rect_conf_uniform}
 and is either top-anchored, bottom-anchored, left-anchored, or right-anchored.
\end{observation}
\begin{proof}
Let $A$ be a maximum-width empty rectangular annulus that satisfies
the condition described in Observation~\ref{obs:rect_conf_uniform}.
Note that $A$ is guaranteed to be uniform.
Now, suppose that $A$ is neither top-anchored, bottom-anchored, left-anchored,
nor right-anchored.
By Observation~\ref{obs:rect_conf_uniform},
at least one side of the inner rectangle of $A$ contains a point of $P$.
Without loss of generality, assume that the top side of the inner rectangle
contains a point of $P$.
Since $A$ is not top-anchored, the top side of the outer rectangle contains
no point of $P$ and is thus at infinity.
This implies that its top-width is unbounded.
Since $A$ is uniform, this should be the case for the other three
bottom-width, left-width, and right-width.
This means that all sides of the outer rectangle of $A$ lie at infinity,
and that no point of $P$ lies on the boundary or outside of the outer rectangle.
This contradicts to the assumption that $A$ is valid.
\end{proof}

Our algorithm will find an empty anchored and uniform rectangular annulus
of maximum width,
which is the correct answer to our problem by Observation~\ref{obs:rect_conf_anchored}.
In the following, we assume without loss of generality that
there exists a maximum-width empty annulus that is uniform and top-anchored,
and describe our algorithm for this case.
The other three cases can be handled analogously.


Let $P = \{p_1, p_2, \ldots, p_n\}$ be the given set of points,
sorted in the descending order of their $y$-coordinates,
that is, $y(p_1) \geq y(p_2) \geq \cdots \geq y(p_n)$.
Consider any empty top-anchored rectangular annulus $A$
that satisfies the condition of Observation~\ref{obs:rect_conf_uniform}.
Let $p_i \in P$ be the point lying on the top side of the outer rectangle of $A$.
By Observation~\ref{obs:rect_conf_uniform}, either
the bottom side of the outer rectangle is at infinity or
there is another point $p_j \in P$ for $i < j \leq n$ on it.
If the bottom side is at infinity, then we say that
a point $p_\infty$ at infinity in the $(-y)$-direction lies on the bottom side.
Thus, in either case, there is $p_j$ on the bottom side of the outer rectangle
for $i < j \leq n$ or $j = \infty$.

Since $A$ is top-anchored, there is a third point $p_k \in P$ on the top side of the inner rectangle of $A$.
Observe that the width of $A$ is determined by the $y$-difference of $p_i$ and $p_k$, that is, $y(p_i) - y(p_k)$.
Thus, the maximum width for top-anchored empty rectangular annuli is one among $O(n^2)$ values $\{y(p_i) - y(p_k) \mid 1 \leq i \leq k \leq n\}$.

The problem becomes even simpler if we fix $p_i$ on the top side of the outer rectangle, since the number of possible widths is reduced to $n$.
An outlook of our algorithm that computes a maximum-width empty top-anchored rectangular annulus is as follows:
(1) For each $p_i \in P$, find an empty annulus $A^*_i$
with $p_i$ lying on the top side of its outer rectangle whose width
is the maximum among the set $\{ y(p_i) - y(p_k) \mid i < k \leq n \}$
and then (2) output the one with maximum width among $A^*_i$ for all $i\in \{1, \ldots, n\}$.
In order to compute $A^*_i$, we try all possible points $p_j$ that bound the bottom side of the outer rectangle.

In the following subsections, we first study the case where two points $p_i$ and $p_j$ on the top and bottom sides are fixed,
and then move on to the case where only a point $p_i$ on the top side is fixed.
More precisely, we discuss a decision algorithm when two points on the top and bottom sides are fixed, and exploit it as a sub-procedure to solve the other case.

\subsection{Decision when two points on top and bottom are fixed}

Suppose that we are given $p_i$ and $p_j$ with $1\leq i+1 < j \leq n$ or $j=\infty$,
and we consider only empty rectangular anuuli whose outer rectangle contains $p_i$ and $p_j$ on its top and bottom sides, respectively.

Here, we consider the following decision problem.
\begin{center}
\noindent\framebox{\begin{minipage}{4.50in}	
\emph{Given}: A positive real $w>0$ \\
\emph{Task}: Does there exist an empty rectangular annulus of width at least $w$ whose outer rectangle contains $p_i$ and $p_j$ on its top and
bottom sides, respectively?
\end{minipage}}
\end{center}

Let $D_{ij}(w)$ denote the outcome of the above decision problem.

\begin{observation} \label{obs:dera}
 If $D_{ij}(w)$ is TRUE, then $D_{ij}(w')$ is TRUE for any $w'\leq w$.
 On the other hand, if $D_{ij}(w)$ is FALSE, then $D_{ij}(w')$ is FALSE for any $w'\geq w$.
\end{observation}


Let $P_{ij} := \{p_{i+1}, \ldots, p_{j-1}\}$ for $i < j\leq n$,
and $P_{i\infty} := \{p_{i+1}, \ldots, p_{n}\}$.
In the following, we show that the decision problem for a given width $w>0$ can be solved by a combination of certain operations on points $P_{ij}$, namely, the \emph{$y$-range $x$-neighbor query} and
the \emph{range maximum-gap query}.
Each of the two operations is described as follows:
\begin{enumerate}[(i)] \denseitems
 \item The \emph{$y$-range $x$-neighbor query}:
 Given three real numbers $(x, y_1, y_2)$,
 this operation is to find two points $q_1$ and $q_2$ in $P_{ij}$ such that
 $q_1$ is the rightmost one among points $P_{ij} \cap [-\infty, x] \times [y_1, y_2]$ and $q_2$ is the leftmost one among points $P_{ij} \cap [x, \infty]\times [y_1, y_2]$.
 Either $q_1$ and $q_2$ may be undefined if there is no point of $P_{i, j}$ in the corresponding range.
 If $q_1$ is undefined, then we return $q_1$ as a point at infinity
 such that $x(q_1) = -\infty$ and $y(q_1) = y_1$;
 if $q_2$ is undefined, then we return $q_2$ such that $x(q_2) = \infty$
 and $y(q_2) = y_1$.
 \item The \emph{range maximum-gap query (in $x$-coordinates)}:
 Given two real numbers $(x_1, x_2)$,
 find the \emph{maximum gap} in the set of real numbers
 $\{x(p) \mid x_1 \leq x(p) \leq x_2, p\in P_{ij}\} \cup \{x_1, x_2\}$,
 where $x(p)$ denotes the $x$-coordinate of point $p$.
 The maximum gap in a set $X$ of real numbers is the maximum difference between two consecutive elements when $X$ is sorted.
 Notice that $x_1$ and $x_2$ are also included in the above set.
 Here, the output of the range maximum-gap query is to be the pair of values
 that define the maximum gap.
\end{enumerate}

 \begin{figure}[tb]
 \centering
 \includegraphics[scale=0.8]{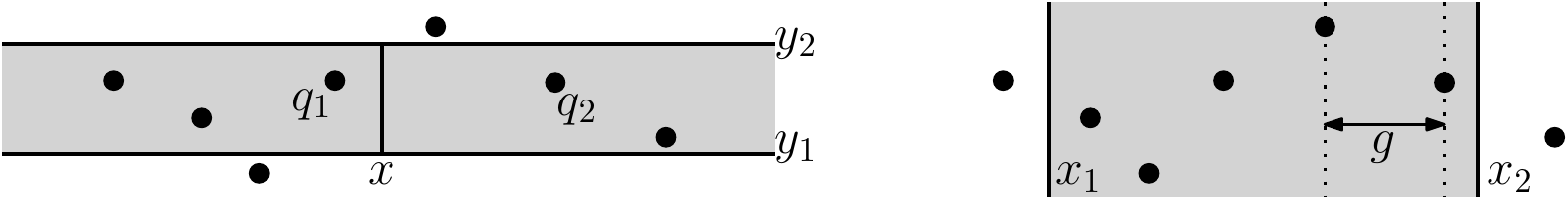}
 \caption{(left) the $y$-range $x$-neighbor query for $(x, y_1, y_2)$ and its answer $q_1$ and $q_2$, (right) the range maximum-gap query for $(x_1, x_2)$ and the maximum gap is $g$. }
 \label{fig:opt_conf}
 \end{figure}

We describe our algorithm for the decision problem
as Algorithm~\ref{alg:ours}.
\begin{algorithm}[tbh]
\KwIn{a width $w>0$}
\KwOut{$D_{ij}(w)$, and an empty rectangular annulus of width $w$ with
$p_i$ and $p_j$ lying on the top and bottom sides of its outer rectangle, respectively, if $D_{ij}(w)$ is TRUE}

 \If {$y(p_i) - y(p_j) < 2w$}
 {Return FALSE.}

 Perform a $y$-range $x$-neighbor query for $(x(p_i), y(p_i) - w, y(p_i))$, and let $q_l$ and $q_r$ be the output with $x(q_l) \leq x(p_i) \leq x(q_r)$.\\

 Perform a $y$-range $x$-neighbor query for $(x(p_j), y(p_j), y(p_j)+w)$, and let $q'_l$ and $q'_r$ be the output with $x(q'_l) \leq x(p_j) \leq x(q'_r)$.\\

 Let $p_l$ be the rightmost one in $\{q_l, q'_l\}$ and
 $p_r$ be the leftmost one in $\{q_r, q'_r\}$.\\

 \If{$\min\{x(p_i), x(p_j)\} < x(p_l)< \max\{x(p_i), x(p_j)\}$ or $\min\{x(p_i), x(p_j)\}<x(p_r)<\max\{x(p_i), x(p_j)\}$}
 {Return FALSE.}

 Perform a range maximum-gap query for $(x(p_l), \min\{x(p_i), x(p_j)\}+w)$, and let $(l, l')$ be the output and $g_l:=l'-l$ be the corresponding maximum gap.\\

 Perform a range maximum-gap query for $(\max\{x(p_i), x(p_j)\}-w, x(p_r))$, and let $(r', r)$ be the output and $g_r:=r-r'$ be the corresponding maximum gap.\\

 \If {$g_l\geq w$ and $g_r\geq w$}
 {Return TRUE, and the rectangular annulus of width $w$ whose outer rectangle is defined by the top and bottom sides through $p_i$ and $p_j$, respectively,
 and the left and right sides at $x = l$ and $x = r$, respectively.}
 \Else{Return FALSE.}

\caption{Decision algorithm}	
\label{alg:ours}
\end{algorithm}

Our decision algorithm, Algorithm~\ref{alg:ours}, evaluates $D_{ij}(w)$
for a given $w$.
As described in Algorithm~\ref{alg:ours},
the decision is made by four calls of the $y$-range $x$-neighbor queries
and the range maximum-gap queries.
Thus, its running time depends on how efficiently we can handle these queries.
If the algorithm decides that $D_{ij}(w)$ is TRUE,
then it also returns a corresponding rectangular annulus,
that is, an empty uniform annulus of width $w$ with $p_i$ and $p_j$
on the top and bottom sides of the outer rectangle.
This can be done by constructing its outer rectangle
since its width $w$ is fixed.

In the following, we show the correctness of our decision algorithm.
\begin{lemma}\label{lem:correct-alg1}
 Algorithm~\ref{alg:ours} correctly computes $D_{ij}(w)$ for any given $w > 0$
 in time $O(T)$, where $T$ is an upper bound on time needed to perform a $y$-range $x$-neighbor query or a range maximum-gap query.
 Moreover, if $D_{ij}(w)$ is TRUE, then an empty rectangular annulus of width $w$ such that $p_i$ and $p_j$ lie on the top and bottom side of its outer rectangle can be found in the same time bound.
\end{lemma}
\begin{proof}
It is straightforward to see that the time complexity is $O(T)$,
since it runs the $y$-range $x$-neighbor queries
and the range maximum-gap queries at most four times.
Hence, we focus on the correctness proof.
Without loss of generality, we assume that $x(p_i)\leq x(p_j)$.

First, we show that if Algorithm~\ref{alg:ours} for input $w$ returns FALSE,
then $D_{ij}(w)$ is FALSE.
There are two cases where Algorithm~\ref{alg:ours} returns FALSE:
Algorithm~\ref{alg:ours} returns FALSE either (1) at line 2, (2) at line 7 or (3) at line 13.
If $y(p_i) - y(p_j) < 2w$, then we cannot build a rectangular annulus of width $w$ with $p_i$ and $p_j$ on the top and bottom sides of its outer rectangle.
Thus, if Algorithm~\ref{alg:ours} returns FALSE either at line 2,
then this is clearly correct.

Suppose case (2), so Algorithm~\ref{alg:ours} returns FALSE at line 7.
From lines 3 and 4, we have the following:
\begin{itemize}
\item $q_1$ is the rightmost point in
 $P_{ij} \cap ([-\infty, x(p_i)]\times[y(p_i)-w, y(p_i)])$.
\item $q_r$ is the leftmost point in
 $P_{ij} \cap ([x(p_i), \infty]\times[y(p_i)-w, y(p_i)])$.
\item $q'_1$ is the rightmost point in $P_{ij} \cap ([-\infty, x(p_j)]\times[y(p_j), y(p_i)+w])$.
\item $q'_r$ is the leftmost point in $P_{ij} \cap ([x(p_j), \infty]\times[y(p_j), y(p_j)+w])$.
\end{itemize}
Also, by line 5, $p_l$ is the rightmost one among $\{q_l, q'_l\}$,
and $p_r$ is the leftmost one among $\{q_r, q'_r\}$.
Then, in this case, we have
\[ x(p_i) < x(p_l) < x(p_j)\quad \text{or} \quad x(p_i)<x(p_r) <x(p_j),\]
by the condition in line 6.
Assume the first condition holds, so $x(p_i) < x(p_l) < x(p_j)$.
This implies that $p_l = q'_l$.
By our construction, we have $y(p_l) = y(q'_l) \in [y(p_j), y(p_j)+w]$.
Consider now any rectangular annulus $A$ of width $w$
whose outer rectangle $R_{out}$
contains $p_i$ and $p_j$ on its top and bottom sides, respectively,
and observe that $A$ must contain the point $p_l$.
Thus, such an annulus $A$ cannot be empty, so $D_{ij}(w)$ is FALSE.
The other case where it holds that $x(p_i)<x(p_r) <x(p_j)$ can be handled
analogously.

Now, suppose case (3), so Algorithm~\ref{alg:ours} returns FALSE at line 13.
Then, by lines 8 and 9, so that we have
\begin{itemize}
\item $(l, l')$ defines the maximum gap $g_l = l'-l$ in range $[x(p_l), x(p_i)+w]$.
\item $(r', r)$ defines the maximum gap $g_r = r-r'$ in range $[x(p_j)-w, x(p_r)]$.
\end{itemize}
In this case, we have
 \[ g_l < w \quad \text{or} \quad g_r < w,\]
by the condition in line 10.
Assume that the first condition holds, so $g_l < w$.
This means that there is no vertical strip of width at least $w$
among points in $P_{ij}$ in the $x$-range $[x(p_l), x(p_i)+w]$.
Consider now any rectangular annulus $A$ of width $w$
whose outer rectangle $R_{out}$
contains $p_i$ and $p_j$ on its top and bottom sides, respectively.
If the left side of $R_{out}$ is to the left of $p_l$, then
$A$ contains $p_l$ by the construction;
otherwise, any vertical strip of width $w$ between $[x(p_l), x(p_i)+w]$
contains a point in $P_{ij}$.
Since the left side of $R_{out}$ cannot be to the right of $p_i$,
this implies that no such annulus $A$ is empty in this case,
so $D_{ij}(w)$ is indeed FALSE.
The other case where it holds that $g_r$ can be handled
analogously.

Next, we show that if Algorithm~\ref{alg:ours} for input $w$ returns TRUE,
then $D_{ij}(w)$ is correctly TRUE.
If Algorithm~\ref{alg:ours} for input $w$ returns TRUE,
then we have the following:
\[ x(p_l) \leq x(p_i), x(p_r) \geq x(p_r),
  g_l \geq w,  g_r \geq w,\]
by the conditions in lines 6 and 10.
Then, we can construct a rectangle $R_{out}$ such that
its top and bottom sides go through $p_i$ and $p_j$, respectively,
and the $x$-coordinates of its left and right sides are $l$ and $r$, respectively.
Let $A$ be the rectangular annulus of width $w$ whose outer rectangle is $R_{out}$.
Then, by our construction, $A$ is empty with respect to $P_{ij}$.
Since $A$ is completely contained in the horizontal strip defined by
the horizontal lines through $p_i$ and $p_j$,
$A$ is also empty with respect to the whole set $P$.
This shows the existence of an empty annulus $A$ of width $w$ such that
the top and bottom sides of its outer rectangle contains $p_i$ and $p_j$,
respectively.
Therefore, we conclude that $D_{ij}(w)$ is TRUE.
In this case, Algorithm~\ref{alg:ours} indeed returns the annuls $A$.
This completes the proof.
\end{proof}

Note that the two operations can be easily done in linear time.
We now show how to perform them in logarithmic time with an aid of the following data structures.
\begin{itemize} \denseitems
 \item Let $\mathcal{D}$ be the data structure on $P$ described
 in Chazelle~\cite{c-asdii-88} that supports a segment dragging query
 for vertical line segments dragged by two horizontal rays.
 A segment dragging query is given by a segment and a direction along two rays
 and is to find the first point in $P$ that is hit by the dragged segment.
 This structure can be constructed in $O(n \log n)$ time using $O(n)$ storage
 \item Let $\mathcal{X}_{ij}$ be a 1D range tree for the $x$-coordinates
 of points in $P_{ij}$
 with an additional field $maxgap(v)$ at each node $v$,
 where $maxgap(v)$ denotes the maximum gap in the canonical subset of $v$.
 Note that $maxgap(v) = 0$ if the canonical subset of $v$ consists of
 only one element.
 The structure $\mathcal{X}_{ij}$ can be constructed
 using storage $O(|P_{ij}|)$~\cite{BCKO08}.
\end{itemize}

Now, suppose that we have already built these two structures
$\mathcal{D}$ and $\mathcal{X}_{ij}$.
Then, the two operations can be handled in $O(\log n)$ time as follows:
\begin{enumerate}[(i)] \denseitems
\item For a $y$-range $x$-neighbor query for $(x, y_1, y_2)$,
we perform two segment dragging queries on $\mathcal{D}$
for a vertical line segment
with endpoints $(x, y_1)$ and $(x, y_2)$ to both the left and the right directions.
These two queries result in
the rightmost point $q_1$ in the range $[-\infty, x] \times [y_1, y_2]$
and the leftmost point $q_2$ in the range $[x, \infty] \times [y_1, y_2]$.
\item For a range maximum-gap query for $(x_1, x_2)$,
we perform a 1D range search for the $x$-range $[x_1, x_2]$
on $\mathcal{X}_{ij}$ again to obtain a collection $C$ of $O(\log n)$ nodes.
The maximum gap in the $x$-coordinates of the points in $P_{ij}$ in the range $[x_1, x_2]$ can be found by comparing $O(\log n)$ values:
$maxgap(v)$ for all $v\in C$ and every gap between two consecutive canonical subset.
\end{enumerate}

Therefore, we conclude the following:
\begin{lemma} \label{lem:decision_with_trees}
 Suppose that we already have two tree structures $\mathcal{D}$
 and $\mathcal{X}_{ij}$.
 Then, Algorithm~\ref{alg:ours} correctly computes $D_{ij}(w)$
 for any given $w > 0$ in time $O(\log n)$.
\end{lemma}

\subsection{Optimization when only a point on top is fixed}
Next, we describe how to find a maximum-width empty top-anchored rectangular annulus such that $p_i$ lies on the top side of the outer rectangle.

Let $w^*_i$ be the width of $A^*_i$.
Observe that $w^*_i$ lies in the set $W_i := \{y(p_i) - y(p_k) \mid i \leq k \leq n\}$.
Instead of solving the optimization problem for each pair $(p_i, p_j)$,
we can rather solve the optimization problem when only a point $p_i$ on top is fixed.
Our algorithm that computes $A^*_i$ and its width $w^*_i$
is presented as in Algorithm~\ref{alg:ours1}.
\begin{algorithm}[t]
\KwIn{A set $P=\{p_1, \ldots, p_n\}$ of points sorted by $y$-coordinates, and a point $p_i \in P$}
\KwOut{$A^*_i$ with $p_i$ and its width $w^*_i$}

Set $k$ to be $i+1$, $w$ to be $0$, and $A$ to be any annulus of width $0$.\\

Build the data structure $\mathcal{D}$ for $P$ and
initialize $\mathcal{X}$ to be $\mathcal{X}_{i,i}$.\\

\For {each $j = i+2, \ldots, n$ and $j = \infty$}
{%
  Insert $p_{j-1}$ if $j \leq n$, or $p_n$ if $j=\infty$,
  into $\mathcal{X}$,
  so that now $\mathcal{X} = \mathcal{X}_{ij}$.\\
  \While {$D_{ij}(y(p_i) - y(p_k))$ is TRUE and $k \leq j$}
  { Set $A$ to be the corresponding annulus of width $y(p_i) - y(p_k)$.\\
    Set $w$ to be $y(p_i) - y(p_k)$.\\
    Increase $k$ by $1$.}
}

Return the current $A$ as $A^*_i$ and the current $w$ as $w^*_i$.
\caption{Computing $A^*_i$.}	
\label{alg:ours1}
\end{algorithm}

\begin{lemma} \label{lem:correct_opt}
 Algorithm~\ref{alg:ours1} can be implemented in $O(n \log n)$ time
 and $O(n)$ space for a fixed $p_i \in P$.
 Also, it correctly computes $w^*_i$ and $A^*_i$.
\end{lemma}
\begin{proof}
In Algorithm~\ref{alg:ours1}, we maintain several variables:
\begin{itemize} \denseitems
\item $A$ is the currently best empty top-anchored annulus with $p_i$
on the top side of its outer rectangle.
\item $w$ is the width of $A$.
\item $k$ is the index such that the next larger width $y(p_i) - y(p_k)$
 is currently being tried.
 It always holds that $w = y(p_i) - y(p_{k-1})$.
\item $\mathcal{X}$ is the 1D range tree with additional field $maxgap(v)$ at each node defined as above on the $x$-coordinates of the points in $P_{ij}$,
 that is, $\mathcal{X} = \mathcal{X}_{ij}$ in the main loop.
\end{itemize}
After initializing these variables properly in lines 1-2,
the main loop runs for $j = i+2, \ldots, n$.
Let $w_{ij}$ be defined as above.
For each $j$, we keep the best annulus we have seen so far as $A$,
not computing the exact value of $w_{ij}$.
That is, we keep
\[ w = \max\{0, w_{i(i+2)}, w_{i(i+3)}, \ldots, w_{ij}\} \]
as the loop invariant for each $j$.
In this way, after the main loop has finished at $j = \infty$, we have
 \[ w = \max_{j=i+2, \ldots, n, \infty} w_{ij} = w^*_i \]
and the corresponding annulus of width $w=w^*_i$ is stored in variable $A$.

Hence, the correctness of Algorithm~\ref{alg:ours1} is guaranteed.
The time complexity is bounded by $O(n \log n)$ as follows.

Lines 1-2 in Algorithm~\ref{alg:ours1} takes $O(n \log n)$ time
since $\mathcal{D}$ can be built in $O(n \log n)$ time~\cite{c-asdii-88}
and $\mathcal{X}_{ii}$ is initialized in $O(1)$ time
as $P_{i,i} = \emptyset$.
In line 4, a point $p_{j-1}$ or $p_n$ is inserted into two tree data structures $\mathcal{T}$ and $\mathcal{X}$.
The structures $\mathcal{X}$ is known to support an insertion in logarithmic time~\cite{BCKO08}.
The additional information $maxgap(v)$ at each node $v$ of $\mathcal{X}$
can be correctly updated in the bottom-up fashion through the path to the root
from the newly inserted node.
Hence, line 4 can be implemented in $O(\log n)$ time.
In line 5, we call the decision algorithm, Algorithm~\ref{alg:ours},
and one execution of line 5 takes $O(\log n)$ time
by Lemma~\ref{lem:decision_with_trees},
since the necessary data structures $\mathcal{D}$ and
$\mathcal{X} = \mathcal{X}_{ij}$ are provided.
Lines 6-8 takes only $O(1)$ time.

To conclude the total running time,
observe that line 4 is executed at most $n-3$ times.
The ``while'' loop in line 5 is executed at most $2(n-3)$ times:
The number of times when the while-condition is false is bounded by $n-3$,
while the number of times when the while-condition is true is also bounded
by $n-3$ since $k$ increases by one whenever this is the case and $k$ cannot be more than $n$.

Consequently, the time complexity of Algorithm~\ref{alg:ours1} is bounded by
$O(n \log n)$.
Since the segment dragging query structure $\mathcal{D}$ and
the 1D range tree $\mathcal{X}$ use $O(n)$ space~\cite{c-asdii-88,BCKO08},
the space usage of Algorithm~\ref{alg:ours1} is bounded by $O(n)$.
\end{proof}

\subsection{Putting it all together}

We are now ready to describe the overall algorithm to solve the MaxERA problem.
Under the assumption that there exists a maximum-width empty rectangular annulus
$A^*$ that satisfies the condition of Observation~\ref{obs:rect_conf_uniform}
and is top-anchored,
we excute Algorithm~\ref{alg:ours1} for each $i=1, \ldots n-1$
and choose the one with the maximum width as $A^*$.
Its correctness is guaranteed by Lemma~\ref{lem:correct_opt}.
The other three cases where there is a maximum-width empty rectangular annulus
that satisfies the condition of Observation~\ref{obs:rect_conf_uniform}
and is either bottom-anchored, left-anchored, or right-anchored,
can be handled in a symmetric way.
Thus, the overall algorithm runs for the four cases
and outputs one with the maximum width.
\begin{theorem} \label{thm:rect}
 Given a set $P$ of $n$ points in the plane,
 a maximum-width rectangular annulus that is empty with respect to $P$
 can be computed in $O(n^2 \log n)$ time and $O(n)$ space.
\end{theorem}
\begin{proof}
The correctness follows from the above discussion
and Observation~\ref{obs:rect_conf_anchored}.

For the time complexity,
we mainly call Algorithm~\ref{alg:ours1} $O(n)$ times.
Thus, it takes $O(n^2 \log n)$ time by Lemma~\ref{lem:correct_opt}.
The space usage can be bounded by $O(n)$ again by Lemma~\ref{lem:correct_opt}.
\end{proof}

\section{Concluding Remarks}
\label{sec:conclusions}
In this paper, we addressed the problem of computing a maximum-width empty
square or rectangular annulus that avoids a given set $P$ of $n$ points in the plane, and presented two efficient algorithms.
Our algorithms run in $O(n^3)$ and $O(n^2 \log n)$ time, respectively, 
for a square and rectangular annulus.
Note that our algorithms are first nontrivial algorithms that solve the problems,
and considered to be efficient compared to the currently best algorithm
for the circular counterpart, which runs in $O(n^3 \log n)$ time~\cite{dhmrs-leap-03}.

There are two obvious open questions.
One asks an improved algorithm with less running time, 
while the other asks a lower bound of the empty annulus problem.
The circular, square, and rectangular versions of the problem are now known to
be solved in $O(n^3 \log n)$, $O(n^3)$, and $O(n^2 \log n)$ time, respectively.
At this moment, it seems difficult to improve each of these upper bounds.
On the other hand, no nontrivial lower bound, other than $\Omega(n)$, 
is known for these problems.
This simply means that nobody fully understand the intrinsic complexity of
this type of problems for now.
It will be very interesting hence if one improves one of these algorithms or proves a nontrivial lower bound for the problem.

{
\bibliographystyle{abbrv}
\bibliography{Ann}
}

\end{document}